\documentclass[11pt]{article}
\usepackage{mathrsfs}
\usepackage{latexsym}
\usepackage{amsmath,amssymb}
\usepackage[pdftex]{hyperref}
\usepackage{amsthm}
\usepackage{amsfonts,cite}
\usepackage[usenames]{color}
\usepackage{amssymb}
\usepackage{graphicx}
\usepackage{amsmath}
\usepackage{amsfonts}
\usepackage{amsthm}
\usepackage{mathrsfs}
\usepackage{dsfont}
\newcommand\be{\begin{equation}}
\newcommand\ee{\end{equation}}
\newcommand\ber{\begin{eqnarray}}
\newcommand\eer{\end{eqnarray}}
\newcommand\berr{\begin{eqnarray*}}
\newcommand\eerr{\end{eqnarray*}}
\newcommand\bea{\begin{eqnarray}}
\newcommand\eea{\end{eqnarray}}
\newcommand\ba{\begin{array}}
\newcommand\ea{\end{array}}
\newcommand\bfR{\mathbb{R}}
\newcommand\dd{\mathrm{d}}
\newcommand\lm{\lambda}
\newcommand\ii{\mbox{i}}
\newcommand\re{\mathrm{e}}

\newcommand\pa{\partial}
\newcommand{\ud}{\mathrm{d}}
\newcommand{\nn}{\nonumber}

\newcommand{\itr}{\int_{\mathbb{R}^2}}
 
\newcommand{\vep}{\varepsilon}
\newcommand{\bi}{\begin{itemize} }
  \newcommand{\ei}{\end{itemize} }

\newcommand\hA{\hat{A}}\newcommand\tA{\tilde{A}}\newcommand\hF{\hat{F}}\newcommand\tF{\tilde{F}}
\newcommand\ov{\overline}

\newtheoremstyle{mythm}{1.5ex plus 1ex minus .2ex}{1.5ex plus 1ex
minus .2ex}{\kai}{\parindent}{\song\bfseries}{}{1em}{}
\numberwithin{equation}{section}\numberwithin{figure}{section}

\newtheorem{theorem}{Theorem}[section]

\allowdisplaybreaks[4]
\begin{document}
\title{Coexisting Vortices and Antivortices \\Generated by Dually Gauged Harmonic Maps}
\author{Xiaosen Han\\Institute of Contemporary Mathematics\\School of Mathematics\\Henan University\\
Kaifeng, Henan 475004, P. R.  China\\ \\
Genggeng Huang\\School of Mathematical Sciences\\
 Fudan University\\ Shanghai 200433, P. R. China\\ \\
Yisong Yang\\
Courant Institute of Mathematical Sciences\\ New York University\\New York, New York 10012, USA}

\date{}
\maketitle

\begin{abstract}
In this paper we first formulate a dually gauged harmonic map model, suggested from a product Abelian Higgs field theory
arising in impurity-inspired field theories, and obtain a new BPS system of equations governing coexisting vortices and antivortices, which are
topologically characterized by the first Chern class of the underlying Hermitian bundle and the Thom class of the associated dual
bundle. We then establish existence and uniqueness theorems for such vortices. For the equations over a compact surface, we
obtain necessary and sufficient conditions for the existence of solutions. For the equations over the full plane, we obtain
all finite-energy solutions. Besides, we also present precise expressions giving the values of various physical quantities of the solutions, including magnetic charges and energies, in terms of the total numbers of vortices and antivortices,
of two species, and the coupling parameters involved.
\end{abstract}

\medskip

\medskip


\section{Introduction}

Vortices in quantum field theory were first conceptualized in the pioneering work of Abrikosov \cite{A} in his prediction of the onset of type-II superconductivity
characterized by the appearance of mixed states, due to the celebrated Meissner effect, in the context of the Ginzburg--Landau theory \cite{GL,T}.
In such a formalism {electromagnetism} is classically governed by a massive Maxwell equation, also known as the London equation \cite{L,T}, but the order parameter
is quantum-mechanically governed by a nonlinear gauged Schr\"{o}dinger  equation, giving rise to a quantum current density to sustain electromagnetism.
Since then vortices have  been realized and recognized broadly in applications and theoretical investigations in areas ranging over
condensed-matter physics, elementary particle physics, and cosmology. Naturally, the richness of applications of vortices has prompted considerable extensions
of the theory beyond the minimally coupled Ginzburg--Landau theory, which are characterized by the presence of multiple scalar and gauge fields introduced to fulfill various theoretical
and phenomenological purposes.
Notably, in the relativistically extended Ginzburg--Landau theory known as the Abelian Higgs theory, vortices appear as
the Nielsen--Olesen strings \cite{NO} which serve to
 mediate in a type-II superconductor the interaction between a monopole
and an anti-monopole resulting in a constant attractive force between the
pair which would confine the monopoles
 \cite{Man1,Man2,Nambu,tH,tH1}. This idea motivated Seiberg and Witten
 \cite{SW} to arrive at a similar mechanism aimed to resolve the quark confinement puzzle. In such an extended setting, numerous
 supersymmetric gauge field theory models are used and the classical Meissner effect is supersymmetrically expanded
so that the Nielsen--Olesen magnetic strings, as well as magnetically charged monopoles, assume the forms of correspondingly revised, colored, counterparts \cite{Auzzi,EF,EFN,EI,GJK,HT,HT2004,MY,ShY2004,ShY-vortex}, to realize a linear confinement picture  \cite{Eto-survey,Gr,Kon-survey,Shifman-survey,ShY2,ShY,Tong}. In these studies
the full vortex, or the complete vortex-monopole complex \cite{C,Ta}, equations are too difficult to analyze. Instead, people have relied on exploring the
underlying, much reduced BPS structure (after the earlier works of Bogomol'nyi \cite{B} and Prasad and
Sommerfeld \cite{PS} on the Yang--Mills--Higgs monopoles and dyons) for the vortex equations.
Besides, in \cite{W}, Witten considered a two-Higgs extended Abelian Higgs
model which serves to generate  cosmic strings as seeds for matter accretion for the galaxy formation in the early universe \cite{K1,K2,V,VS}; in \cite{Babaev},
Babaev studied a two-flavor Ginzburg--Landau theory aimed at modeling two-gap superconductivity which gives rise to fractionally magnetized vortices absent
in conventional single-gap situations; in \cite{D,IS}, some two-Higgs particle extensions of the Abelian Higgs theory are used to describe double-layer
fractional quantum Hall effect in terms of the Chern--Simons kinetics. These and other applications have led to some active research on vortices generated in
quantum-field theory models accommodating extended gauge and matter field dynamics.
Motivated by these studies, in the present work, we consider coexisting vortices and antivortices, carrying opposite magnetic charges,
arising in a field theory containing two Higgs scalar fields generated from two gauged harmonic maps. This problem owes its origin from several subjects of
distinguished interest
and significance in field theories: Firstly, it originates from the classical integrable sigma model studied by Belavin and Polyakov \cite{BP} where the configuration map
is the spin vector describing the magnetic orientation in a ferromagnet which is mathematically the simplest harmonic map of a nontrivial topological characterization
\cite{ES}. Secondly, its gauge-theoretical content was initially explored by Schroers \cite{Sch1,Sch2} to host electromagnetism, whose elegant BPS structure enabled
 an Abelian Higgs theory \cite{Y1,Y2},  in which vortices and antivortices of opposite magnetic charges coexist, to be developed.
Thirdly, and more recently,
 electric and magnetic impurities are considered in the Abelian Higgs model in the context of
supersymmetric field theories and the usual BPS structure is shown to be preserved in the presence of such impurities \cite{HKT}.
In particular, in \cite{TW}, Tong and Wong proposed that magnetic impurities may be viewed as heavy, frozen vortices sitting in an additional
Abelian gauge group, so that the interaction of the Abelian Higgs vortices with impurities may be described in the
framework of a product
Abelian gauge field theory with two scalar fields, which was  later shown \cite{HY2} to enjoy a general product Abelian gauge-field-theory formalism,
allowing an extension to include the Chern--Simons dynamics as well.  Inspired by these studies, we shall develop in the present work a product Abelian gauge field theory which accommodates four species of oppositely charged and multiply distributed BPS vortices induced from two Higgs fields. Specifically a solution would possess two species
of positively charged vortices of the vortex numbers $N_1, N_2$ and two species of negatively charged vortices of the vortex numbers $P_1,P_2$, respectively.
We will present existence and uniqueness theorems for such vortex solutions under necessary and sufficient conditions given explicitly in terms of $N_1, N_2, P_1,P_2$, and other physical parameters in our field-theoretical framework. The rich properties of these solution configurations may be useful in offering broader vortex
phenomenologies in quantum field theories, in view of \cite{Babaev,D,Eto-survey,Gr,IS,K1,K2,Kon-survey,Shifman-survey,ShY2,ShY,Tong,V,VS,W,Y1,Y2}, for example, and elsewhere, and stimulate further exploration.

The rest of the paper is organized as follows. In the next section, we first review the bare \cite{BP,R,RS} and gauged
\cite{Sch1,Sch2,SSY,Y1,Y2} harmonic models  for the purpose of illustrating how vortices of opposite vortex charges arise and what new
properties are to be expected. We then present a dually gauged harmonic map model hosting two interacting harmonic maps and gauge
fields along the line of a product Abelian Higgs theory accommodating impurities \cite{HKT,TW}. We will derive a new
BPS system of equations and demonstrate how two species of vortices and antivortices arise. We will also show how this system of
equations reduces into the system that arises in the product Abelian Higgs theory \cite{HY} recently uncovered to extend the
formalism in \cite{TW}. In Section 3, we state our existence and uniqueness theorems for solutions of the BPS systems of equations over a compact surface and on the full plane. Our mathematical analysis is based on calculus of variations and elliptic {\em a priori} estimates. Specifically, in Section 4, we prove the theorem in the compact-surface situation.
 Technically,  the governing functional assumes a logarithmic form which makes the underlying analytic structure more difficult from
those already investigated in the literature. In our situation, fortunately, we encounter two logarithmic terms which may be seen
to compensate each other in such a way that, jointly, they give rise to a linear lower bound, thus enabling a resolution to the logarithmic
difficulty.
 In Section 5,
we establish the theorem in the full-plane situation.
In this situation, it is difficult to get the coerciveness of the
associated functional straightforwardly in the  usual Sobolev space $H^1(\mathbb{R}^2)$, due to the logarithmic nonlinear terms again, coupled with the issue associated with loss of compactness. To overcome this difficulty, we
use some estimates which involve a combination of the $L^2$-norm  and  the $L^1$-norm of the
minimizing sequence in two different domains,
so as to achieve a desired energy control of the sequence. In our approach,  we show how the difficulty may be resolved by
our domain splitting method and the use of a special form of
the Gagliardo--Nirenberg interpolation inequality so that a minimizing sequence is eventually shown to be
bounded in $H^1(\bfR^2)$. Thus a minimizer may be obtained as a weak limit of the sequence.   These technical novelties lead us to a complete understanding of
the problem of existence and uniqueness of solutions realizing two species of prescribed vortices and antivortices in the proposed dually coupled
harmonic map model. In Section 6, we derive sharp exponential decay estimates and quantized integrals for a planar solution. In Section 7, we briefly summarize our work and make some comments.

\section{Vortex equations induced from gauged harmonic maps}

In this section, we aim to derive a dually gauged harmonic map field theory which allows the coexistence of vortices and
antivortices. In order to {motivate}  the derivation, we begin by a discussion of the gauged harmonic map model, especially its
{origins} from the classical harmonic map model. We then derive the dually gauged theory and show how vortices and antivortices
arise. We end the section to comment on a natural link of the solutions to harmonic maps.

\subsection{Classical and gauged harmonic map models}

Recall that in the classical (static)
$O(3)$ sigma model describing
a planar ferromagnet
the field configuration is a spin vector $\phi=(\phi_1,\phi_2,\phi_3)$
which maps $\bfR^2$ into the unit sphere,
$S^2$, in $\bfR^3$, namely, $\phi_1^2+\phi_2^2+\phi_3^2=1$.
The energy then reads \cite{BP,R,RS}
\be\label{j2.1}
E(\phi)=\frac12\int_{\bfR^2}\left\{ (\pa_1\phi)^2+(\pa_2\phi)^2\right\}\,\dd x.
\ee
Due to the finite-energy requirement, we may assume $\phi$ approaches a fixed vector at infinity. Thus, for $\phi$, we may
compactify $\bfR^2$ so that
$\phi$ belongs to a second homotopy class on $S^2$ characterized by
its corresponding Brouwer\index{Brouwer degree} degree, $\deg(\phi)$,
which may be represented by the following normalized area integral,
\be\label{j2.2}
\deg(\phi)=\frac1{4\pi}\int_{\bfR^2}\phi\cdot(\pa_1\phi\times\pa_2\phi)\,\dd x.
\ee
In view of (\ref{j2.2}),
it is seen \cite{BP,R,RS} that there holds the topological energy bound
\be\label{j2.3}
E(\phi)\geq 4\pi|\deg(\phi)|,
\ee
and that, in each $\deg(\phi)=N$ class, solutions saturating the energy lower
 bound (\ref{j2.3}) could be
constructed explicitly via meromorphic functions \cite{BP,R}.
On the other hand, in the gauged $O(3)$ sigma model proposed in the work of Schroers \cite{Sch1},
the energy (\ref{j2.1}) is extended to take
the form
\be\label{j2.5}
E(\phi,A)=\frac12\int_{\bfR^2}\left\{ F_{12}^2+(D_1\phi)^2+(D_2\phi)^2+(1-{\bf n}\cdot\phi)^2\right\}\,\dd x,
\ee
where $A_i$ ($i=1,2$) is a vector field, $F_{12}=\pa_1 A_2-\pa_2 A_1$  the induced magnetic
curvature,  $D_i\phi=\pa_i\phi-A_i({\bf n}\times\phi), i=1,2,$ are covariant derivatives, and ${\bf n}=(0,0,1)$ is the north pole on $S^2$, and
the degree formula (\ref{j2.2}) may be modified to assume the form
\be\label{j2.6}
\deg(\phi)=\frac1{4\pi}\int_{\bfR^2}\left\{\phi\cdot(D_1\phi\times D_2\phi)-F_{12}(1-{\bf n}\cdot\phi)\right\}\,\dd x,
\ee
so that the same  energy
 bound (\ref{j2.3}) may be established. Here we emphasize that
this construction relies on the structure of the energy (\ref{j2.5}) which specifies a fixed groundstate, $\phi={\bf n}$, at the
infinity of $\bfR^2$, enabling its compactification into $S^2$ as before, thus the validity of the degree formula (\ref{j2.6}). Of course this
form of the energy
breaks the original $O(3)$ symmetry in the bare energy (\ref{j2.1}). To exploit the broken symmetry and explore the electromagnetism induced
from the gauge field, in \cite{Sch2}, Schroers revisited the gauged sigma model and formulated the theory into a general setting,
whose energy  essentially  assumes the form
\be\label{jEE}
E(\phi,A)=\frac12\int_{\bfR^2}\left\{ F_{12}^2+(D_1\phi)^2+(D_2\phi)^2+ (\sigma-{\bf n}\cdot\phi)^2\right\}\,\dd x,
\ee
where $\sigma\in[0,1]$ determines the angle between the north pole ${\bf n}$ and $\phi$ at infinity in $S^2$. In
particular, when $\sigma<1$, the set of groundstates, or vacua, becomes a  circle  manifold defined by the equations
\be
\phi_1^2+\phi^2_2=1-\sigma^2 ,\quad\phi_3=\sigma,
\ee
so that the theory possesses
 a spontaneously broken symmetry which leads to the appearance of
vortices and
antivortices, as in the Ginzburg--Landau theory \cite{A,GL}. As a consequence,
  a new Abelian Higgs theory naturally arises. In fact, without loss of generality, take $\sigma=0$ and consider a complex scalar
field $u=u_1+\ii u_2$ induced from the map $\phi$ so that
\be\label{j2.8}
u_1=\frac{\phi_1}{1+\phi_3},\quad u_2=\frac{\phi_2}{1+\phi_3}.
\ee
That is,
 we project $S^2$ onto the complex plane through the south pole $-{\bf n}$, which
corresponds to infinity of $u$.
Thus, with the induced gauge-covariant
derivatives $D_i u=\pa_i u-\ii A_i u$ ($i=1,2$),
the normalized energy density given in (\ref{jEE}) becomes
\be\label{j2.9}
{\cal H}=\frac12\,F_{12}^2+\frac2{(1+|u|^2)^2}\sum_{i=1}^2
(D_i u)
\overline{(D_i u)}
+\frac12\bigg(\frac{1-|u|^2}{1+|u|^2}\bigg)^2.
\ee

There are two interesting and relevant facts worthy noticing.

(i) Like that in the classical Yang--Mills--Higgs theory,
the potential density function for the complex scalar field $u$ also has
a Mexican-hat profile.
In particular, when we take the $|u|^2\to1$ limit in the denominators of the second
and third terms in (\ref{j2.9}), we see that the model approaches that of the classical Abelian
Higgs theory,
\be\label{j2.10}
{\cal H}=\frac12\,F_{12}^2+\frac12\sum_{i=1}^2
(D_i u)
\overline{(D_i u)}
+\frac18({1-|u|^2})^2,
\ee
for which an existence and uniqueness theorem for multiply distributed vortices was established over $\bfR^2$ in \cite{JT,T1,T2} and over a compact surface $S$
in \cite{Br,N1,N2,WY} where it is shown that the total vortex number $N$ needs to satisfy the bound
\be\label{jB}
N<\frac{|S|}{4\pi},
\ee
in which $|S|$ being the surface area of $S$, to ensure existence of a BPS solution. The bound given in (\ref{jB}) is sometimes referred to as the Bradlow bound \cite{Adam,MN,Nasir}.

(ii) The preimages of the north pole under the original spin vector $\phi$ become the zeros
of the complex field $u$ and those of the south pole the poles of $u$.
Moreover, the energy density (\ref{j2.9}) is invariant under the transformation
\be\label{j2.11}
(u, A_\mu)\mapsto \bigg(\frac1u,-A_\mu\bigg)
\ee
in addition to its $U(1)$ gauge invariance. This important feature indicates that the poles
and zeros will play equal roles. Specifically, the magnetic field $F_{12}$ may be shown to be governed in the BPS limit of the equations of motion of (\ref{j2.9}) by the formula
\cite{Sch1,Y1,Y2}:
\be\label{jF}
F_{12}=\frac{1-|u|^2}{1+|u|^2},
\ee
so that the zeros and poles of $u$ give rise to vortices with $F_{12}=1$ and antivortices with $F_{12}=-1$, respectively, and that the total energy reads
\be\label{jE}
E=\int{\cal H}\,\dd x=2\pi (N+P),
\ee
where $N,P$ are the numbers (counting algebraic multiplicities) of zeros and poles of $u$, which are also the total vortex and antivortex numbers of the system.
Furthermore, the Bradlow bound (\ref{jB}) is now replaced with the updated bound \cite{SSY}
\be\label{jNP}
|N-P|<\frac{|S|}{2\pi},
\ee
which implies that the total vortex number, $N+P$, and thus the energy as well, as given in (\ref{jE}), may be arbitrarily high, so far as the discrepancy of the two types of the vortices, measured by the quantity $|N-P|$, remains under control by (\ref{jNP}).

\subsection{Dually gauged harmonic map model with two interacting configuration maps}

We are now prepared to consider a dually gauged harmonic map (or sigma) model, with two interacting configuration maps, $\phi,\psi$, with images
in $S^2$, and two Abelian gauge fields,
$\hat{A}_i,\tilde{A}_i$, in the static situation (for simplicity). As before, use $\hat{F}_{12}$ and $\tilde{F}_{12}$ to denote the magnetic fields induced from
$\hat{A}_i$ and $\tilde{A}_i$, respectively. Let
\be
D_i\phi=\pa_i\phi-(a \hat{A}_i+b\tilde{A}_i)({\bf n}\times\phi),\quad D_i\psi=\pa_i\psi-(c\hat{A}_i+d\tilde{A}_i)({\bf n}\times \psi),
\ee
be the covariant derivatives associated with the charge parameters $a,b,c,d$. It can be seen that there holds the identity
\be
\phi\cdot(\pa_1\phi\times\pa_2\phi)=\phi\cdot(D_1\phi\times D_2\phi)
+(a\hat{F}_{12}+b\tilde{F}_{12})({\bf n}\cdot\phi-1)-a\hat{H}_{12}-b\tilde{H}_{12},
\ee
where $\hat{H}_{12}=\pa_1\hat{H}_2-\pa_2\hat{H}_1$ (the same for $\tilde{H}_{12}$) with
\be
\hat{H}_i=\hat{A}_i({\bf n}\cdot\phi-1),\quad \tilde{H}_i=\tilde{A}_i({\bf n}\cdot\phi-1).
\ee
A similar expression for $\psi\cdot(\pa_1\psi\times\pa_2\psi)$ also holds. Hence we have the degree formula
\be
\deg(\phi)=\frac1{4\pi}\int_{\bfR^2}\left\{\phi\cdot(D_1\phi\times D_2\phi)+(a\hat{F}_{12}+b\tilde{F}_{12})({\bf n}\cdot\phi-1)\right\}\,\dd x,
\ee
provided that $\phi={\bf n}$ at infinity of $\bfR^2$.
A similar expression holds for $\deg(\psi)$. Hence, combining these facts with the methods in \cite{B,PS,Sch1}, we arrive at the
BPS energy density
\bea\label{jH0}
{\cal H}&=&\frac12\hat{F}^2_{12}+\frac12\tilde{F}_{12}^2+\sum_{i=1}^2\left(|D_i\phi|^2+|D_i\psi|^2\right)\nn\\
&&+2\left(a[{\bf n}\cdot\phi-1]+c[{\bf n}\cdot\psi-1]\right)^2+2\left(b[{\bf n}\cdot\phi-1]+d[{\bf n}\cdot\psi-1]\right)^2,
\eea
which leads to the following extended potential density function in view of (\ref{jEE}):
\be
V(\phi,\psi)=2\left(a[{\bf n}\cdot\phi-1]+c[{\bf n}\cdot\psi-1]+\xi\right)^2+2\left(b[{\bf n}\cdot\phi-1]+d[{\bf n}\cdot\psi-1]+\gamma\right)^2,
\ee
where $\xi,\gamma$ are coupling parameters. Thus, corresponding to the symmetric vacuum state case,
\be
\phi=\psi={\bf n},
\ee
at infinity, we have $\xi=\gamma=0$, and to the spontaneously broken {symmetry} case,
\be
{\bf n}\cdot\phi={\bf n}\cdot\psi=0,
\ee
at infinity, we have
\be
\xi=a+c,\quad\gamma=b+d,
\ee
similar to the case $\sigma=0$ in (\ref{jEE}). Therefore, in this latter case,
we can write down the normalized energy density governing the scalar fields $\phi,\psi$ coupled with the gauge fields $\hat{A}_i,
\tilde{A}_i$ as
\bea\label{jH}
{\cal H}&=&\frac12\hat{F}^2_{12}+\frac12\tilde{F}_{12}^2+\sum_{i=1}^2\left(|D_i\phi|^2+|D_i\psi|^2\right)\nn\\
&&+2\left(a{\bf n}\cdot\phi+c{\bf n}\cdot\psi\right)^2+2\left(b{\bf n}\cdot\phi+d{\bf n}\cdot\psi\right)^2.
\eea

In order to illustrate the topological structure of the model more transparently, we now represent the maps $\phi,\psi$ by a pair of complex scalar fields $q,p$,
via formulas like (\ref{j2.8}) leading to the relation
\be
\phi=\left(\frac{2\Re{(q)}}{1+|q|^2},\frac{2\Im{(q)}}{1+|q|^2},\frac{1-|q|^2}{1+|q|^2}\right),
\ee
between $\phi$ and $q$, and similarly for $\psi$ and $p$,
realized as two cross sections over a Hermitian line bundle $L$ over a Riemann surface $S$, either compact or non-compact, and the gauge fields $\hat{A}_i,\tilde{A}_i$ as
 two connection 1-forms $\hA,\tA$ which induce the magnetic fields as curvature 2-forms $\hF=\dd \hA,\tF=\dd\tA$, with the connections
\be
Dq=\dd q-\ii (a\hA+b\tA)q,\quad Dp=\dd p-\ii (c\hA+d\tA)p,
\ee
operating on $q,p$, respectively, where $a,b,c,d$ are seen to be real coupling parameters, whose roles are to mix the interaction of $q,p,\hA,\tA$.
Thus, with this notation, the energy density (\ref{jH}) becomes
\bea\label{2}
{\cal H}&=&\frac12*(\hF\wedge *\hF)+\frac12*(\tF\wedge *\tF)\nn\\
&&+\frac4{(1+|q|^2)^2}*(Dq\wedge * \overline{Dq})+\frac4{(1+|p|^2)^2}*(Dp\wedge * \overline{Dp})\nn\\
&&+2\left(a\left[\frac{1-|q|^2}{1+|q|^2}\right]+c \left[\frac{1-|p|^2}{1+|p|^2}\right]\right)^2+2\left(b\left[\frac{1-|q|^2}{1+|q|^2}\right]+d \left[\frac{1-|p|^2}{1+|p|^2}\right]\right)^2,
\eea
where $*$ is the Hodge dual. So it follows that the Euler--Lagrange equations of (\ref{2}) are
\bea
\frac12\, D*\left(\frac{Dq}{(1+|q|^2)^2}\right)&=&-\frac{*(Dq\wedge *\overline{Dq})}{(1+|q|^2)^3}q\nn\\
&&-\frac{q}{(1+|q|^2)^2}\left([a^2+b^2]\left[\frac{1-|q|^2}{1+|q|^2}\right]
+[ac+bd]\left[\frac{1-|p|^2}{1+|p|^2}\right]\right),\nn\\
&&\label{2.19}\\
\frac12\, D*\left(\frac{Dp}{(1+|p|^2)^2}\right)&=&-\frac{*(Dp\wedge *\overline{Dp})}{(1+|p|^2)^3}p\nn\\
&&-\frac{p}{(1+|p|^2)^2}\left([ac+bd]\left[\frac{1-|q|^2}{1+|q|^2}\right]
+[c^2+d^2]\left[\frac{1-|p|^2}{1+|p|^2}\right]\right),\nn\\
&&\label{2.20}\\
\frac14\,\dd *\hat{F}&=&\frac{a\ii (\overline{q}Dq-q\overline{Dq})}{(1+|q|^2)^2}+\frac{c\ii (\overline{p}Dp-p\overline{Dp})}{(1+|p|^2)^2},\label{2.21}\\
\frac14\,\dd *\tilde{F}&=&\frac{b\ii (\overline{q}Dq-q\overline{Dq})}{(1+|q|^2)^2}+\frac{d\ii (\overline{p}Dp-p\overline{Dp})}{(1+|p|^2)^2},\label{2.22}
\eea
which are rather complicated.
It is interesting to note that, in the limit $|q|^2,|p|^2\to 1$, (\ref{2}) becomes
\bea\label{3}
{\cal H}&=&\frac12*(\hF\wedge *\hF)+\frac12*(\tF\wedge *\tF)+
*(Dq\wedge * \overline{Dq})+*(Dp\wedge * \overline{Dp})\nn\\
&&+\frac12\left(a\left[{1-|q|^2}\right]+c \left[{1-|p|^2}\right]\right)^2+\frac12\left(b\left[{1-|q|^2}\right]+d \left[{1-|p|^2}\right]\right)^2,
\eea
which has been studied in \cite{HY2} so that the Abelian Higgs theory with impurity studied in \cite{TW} corresponds to the choice $a=1,b=-1,c=0,d=1$.
For (\ref{3}), the Euler--Lagrange equations, or the generalized two-gap Ginzburg--Landau equations in our context, are
\bea
D*Dq&=&-\left([a^2+b^2][1-|q|^2]+[ac+bd][1-|p|^2]\right)q,\label{2.24}\\
 D*Dp&=&-\left([ac+bd][1-|q|^2]+[c^2+d^2][1-|p|^2]\right)p,\label{2.25}\\
\dd *\hat{F}&=&\ii a(\overline{q} Dq-q\overline{Dq})+\ii c(\overline{p} Dp-p\overline{Dp}),\label{2.26}\\
\dd *\tilde{F}&=&\ii b(\overline{q} Dq-q\overline{Dq})+\ii d(\overline{p} Dp-p\overline{Dp}).\label{2.27}
\eea
Setting $|q|^2=1, |p|^2=1$ in the denominators in (\ref{2}), (\ref{2.19})--(\ref{2.22}), we can recover (\ref{3}), (\ref{2.24})--(\ref{2.27}), respectively.

We now pursue a BPS reduction to the full governing equations (\ref{2.19})--(\ref{2.22}). To this goal, recall the identities
\bea
Dq\wedge *\overline{Dq}+*Dq\wedge \ov{Dq}&=&(Dq\pm\ii * Dq)\wedge *\ov{(Dq\pm\ii * Dq)}\nn\\
&&\pm\ii (Dq\wedge\ov{Dq}-*Dq\wedge *\ov{Dq}),
\eea
\be
|Dq|^2=*(Dq\wedge*\ov{Dq}),
\ee
etc. Besides, for the current densities
\bea
J(q)&=&\frac{\ii}{1+|q|^2}\left(q\ov{Dq}-\ov{q}Dq\right),\\
J(p)&=&\frac{\ii}{1+|p|^2}\left(p\ov{Dp}-\ov{p}Dp\right),
\eea
we have
\bea
K(q)&\equiv& \dd J(q)\nn\\
&=&-\frac{2|q|^2}{1+|q|^2} (a\hF+b\tF)+\frac{\ii}{(1+|q|^2)^2}\left(Dq\wedge \ov{Dq}-*Dq\wedge *\ov{Dq}\right),\label{8}\\
K(p)&\equiv& \dd J(p)\nn\\
&=&-\frac{2|p|^2}{1+|p|^2} (c\hF+d\tF)+\frac{\ii}{(1+|p|^2)^2}\left(Dp\wedge \ov{Dp}-*Dp\wedge *\ov{Dp}\right).\label{9}
\eea
Hence, we obtain the energy decomposition
\bea\label{10}
{\cal H}&=&\frac12\left|\hF \mp 2* \left(a\left[\frac{1-|q|^2}{1+|q|^2}\right]+c \left[\frac{1-|p|^2}{1+|p|^2}\right]\right)\right|^2\nn\\
&&+\frac12\left|\tF\mp  2*\left(b\left[\frac{1-|q|^2}{1+|q|^2}\right]+d \left[\frac{1-|p|^2}{1+|p|^2}\right]\right)\right|^2\nn\\
&&\pm *2 \hF\left(a\left[\frac{1-|q|^2}{1+|q|^2}\right]+c \left[\frac{1-|p|^2}{1+|p|^2}\right]\right)
\pm *2\tF \left(b\left[\frac{1-|q|^2}{1+|q|^2}\right]+d \left[\frac{1-|p|^2}{1+|p|^2}\right]\right)\nn\\
&&+\frac2{(1+|q|^2)^2}\left(|Dq\pm\ii * Dq|^2\pm\ii *(Dq\wedge \ov{Dq}-*Dq\wedge *\ov{Dq})\right)\nn\\
&&+\frac2{(1+|p|^2)^2}\left(|Dp\pm\ii * Dp|^2\pm\ii *(Dp\wedge \ov{Dp}-*Dp\wedge *\ov{Dp})\right)\nn\\
&=&\frac12\left|\hF \mp 2* \left(a\left[\frac{1-|q|^2}{1+|q|^2}\right]+c \left[\frac{1-|p|^2}{1+|p|^2}\right]\right)\right|^2\nn\\
&&+\frac12\left|\tF\mp  2*\left(b\left[\frac{1-|q|^2}{1+|q|^2}\right]+d \left[\frac{1-|p|^2}{1+|p|^2}\right]\right)\right|^2\nn\\
&&+\frac2{(1+|q|^2)^2}|Dq\pm\ii * Dq|^2+\frac2{(1+|p|^2)^2}|Dp\pm\ii * Dp|^2\nn\\
&&\pm *2(a\hF+b\tF)\pm*2\left(-\frac{2|q|^2}{1+|q|^2}(a\hF+b\tF)\right)\nn\\
&&\pm *2(c\hF+d\tF)\pm*2\left(-\frac{2|p|^2}{1+|p|^2}(c\hF+d\tF)\right)\nn\\
&&\pm\frac{2\ii}{(1+|q|^2)^2} *\left(Dq\wedge \ov{Dq}-*Dq\wedge *\ov{Dq}\right)\nn\\
&&\pm\frac{2\ii}{(1+|p|^2)^2} *\left(Dp\wedge \ov{Dp}-*Dp\wedge *\ov{Dp}\right).
\eea

Applying (\ref{8}), (\ref{9}) in (\ref{10}), we arrive at the neat expression
\bea\label{11}
{\cal H}&=&\frac12\left|\hF \mp 2* \left(a\left[\frac{1-|q|^2}{1+|q|^2}\right]+c \left[\frac{1-|p|^2}{1+|p|^2}\right]\right)\right|^2\nn\\
&&+\frac12\left|\tF\mp  2*\left(b\left[\frac{1-|q|^2}{1+|q|^2}\right]+d \left[\frac{1-|p|^2}{1+|p|^2}\right]\right)\right|^2\nn\\
&&+\frac2{(1+|q|^2)^2}|Dq\pm\ii * Dq|^2+\frac2{(1+|p|^2)^2}|Dp\pm\ii * Dp|^2\nn\\
&&\pm 2*\left((a+c) \hF+(b+d)\tF+K(q)+K(p)\right).
\eea

Observe that the quantities
\be
\frac1{2\pi}\int_S \hF, \quad\frac1{2\pi}\int_S \tF, \label{2.46}
\ee
are related to the first Chern classes represented by the curvature 2-forms $\hF,\tF$, respectively, and
\be
\frac1{4\pi}\int_S K(q),\quad\frac1{4\pi}\int_S K(p), \label{2.47}
\ee
the Thom classes of the dual bundle of $L$, represented by the mixed gauge connections $a\hA+b\tA,c\hA+d\tA$, respectively. For some detailed computation
and characterization of these topological invariants, see \cite{SSY}. Thus
\be
\tau\equiv 2\left((a+c) \hF+(b+d)\tF+K(q)+K(p)\right) \label{2.48}
\ee
is a topological density which yields via (\ref{11}) the topological lower bound
\be
E=\int_S{\cal E}*1\geq\left|\int_S \tau\right|, \label{2.49}
\ee
which is attained when $(q,p,\hA,\tA)$ satisfies the BPS equations
\bea
Dq\pm\ii *Dq&=&0,\label{16}\\
Dp\pm\ii *Dp&=&0,\label{17}\\
\hF&=&\pm 2*\left(a\left[\frac{1-|q|^2}{1+|q|^2}\right]+c \left[\frac{1-|p|^2}{1+|p|^2}\right]\right),\label{18}\\
\tF&=&\pm 2*\left(b\left[\frac{1-|q|^2}{1+|q|^2}\right]+d \left[\frac{1-|p|^2}{1+|p|^2}\right]\right).\label{19}
\eea
Here and in the sequel we observe the convention that  we  choose either the upper or lower sign in all equations simultaneously.
It is straightforward to examine that (\ref{16})--(\ref{19}) imply (\ref{2.19})--(\ref{2.22}). Therefore, we have arrived at a significant reduction of the complicated
system of equations (\ref{2.19})--(\ref{2.22}) into the system of equations (\ref{16})--(\ref{19}), along the spirit of Bogomol'nyi \cite{B} and Prasad--Sommerfield
\cite{PS}, which will be the focus of our study to follow.

We next explain how vortices and antivortices arise in the system of the BPS equations  (\ref{16})--(\ref{19}). First, the equations (\ref{16})--(\ref{17}) indicate
that $q,p$ are meromorphic so that their zeros and poles are isolated and of integer multiplicities. Thus, counting multiplicities, we may let the sets of zeros and poles of $q,p$ be denoted by
\bea
&& {\cal Z}(q)=\{z'_{1,1},\dots,z'_{1,N_1}\},\quad {\cal P}(q)=\{z''_{1,1},\dots,z''_{1,P_1}\},\label{20}\\
&&{\cal Z}(p)=\{z'_{2,1},\dots,z'_{2,N_2}\},\quad {\cal P}(p)=\{z''_{2,1},\dots,z''_{2,P_2}\},\label{21}
\eea
respectively. That is, algebraically, $q,p$ have $N_1,N_2$ zeros and $P_1,P_2$ poles, respectively, as indicated. Then introduce
\be
B_1=d*\hat{F}-c*\tilde{F},\quad B_2=a*\tilde{F}-b*\hat{F},
\ee
as two induced magnetic fields. Consequently, in view of (\ref{18})--(\ref{19}), we obtain
\bea
B_1&=&\pm2(ad-bc)\left(\frac{1-|q|^2}{1+|q|^2}\right),\label{2.57}\\
B_2&=&\pm2(ad-bc)\left(\frac{1-|p|^2}{1+|p|^2}\right),\label{2.58}
\eea
so that vortices and antivortices are exhibited and presented by the zeros and poles of $q,p$ clearly, where $B_1,B_2$ attain their global maximum and minimum, $\pm 2(ad-bc)$, depending on the choice of signs, respectively.

Later, we will obtain the quantities
\bea
\frac1{2\pi}\int_S \hF&=&\frac{(d[N_1-P_1]-b[N_2-P_2])}{ad-bc}, \\
\frac1{2\pi}\int_S \tF&=&\frac{(a[N_2-P_2]-c[N_1-P_1])}{ad-bc},
\eea
which give rise to the associated first Chern classes
\bea
\frac1{2\pi}\int_S{ F}(q)&=&N_1-P_1,\\
\frac1{2\pi}\int_S{ F}(p)&=&N_2-P_2,
\eea
where ${ F}(q)=\dd A(q)$ and ${ F}(p)=\dd A(p)$ are the curvature 2-forms induced from the connection 1-forms
\be
 { A}(q)=a\hat{A}+b\tilde{A},\quad { A}(p)=c\hat{A}+d\tilde{A},
\ee
which take account of the the differences of the numbers of zeros and poles, of the sections $q,p$, respectively, and solely.
As another consequence, we are led to the following quantized values of the magnetic charges or fluxes
\bea
\int_SB_1\,\dd \Omega_g &=&\frac{2\pi\big([c^2+d^2][N_1-P_1]-[ac+bd][N_2-P_2]\big)}{ad-bc},\label{2.64}\\
\int_SB_2\,\dd\Omega_g&=&\frac{2\pi\big([a^2+b^2][N_2-P_2]-[ac+bd][N_1-P_1]\big)}{ad-bc},\label{2.65}
\eea
which depend on the full spectrum of the numbers of the zeros and poles of both $q$ and $p$.
\medskip

As an illustration, in  Figure \ref{F1}, we present a plot of the strength of one of the two identified magnetic fields which peaks and valleys at the centers of one of the two species of vortices and antivortices
represented by the zeros and poles of one of the two Higgs scalar fields descending from two coupled and gauged harmonic maps, respectively.

\begin{figure}[h]
\begin{center}
\includegraphics[height=6cm,width=9cm]{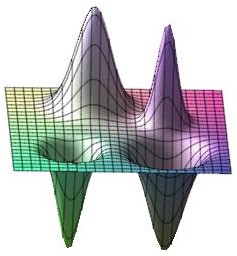}
\caption{A computer-generated plot of the planar distribution of one of the two  magnetic fields induced from one of the two pairs of vortices and antivortices of varied local vortex charges.
It is seen that the magnetic field is locally concentrated and decays fast away from the vortex cores. }
\label{F1}
\end{center}
\end{figure}

\subsection{Notes on the analytic properties of solutions}

Here we comment briefly on the analytic properties of the solution maps so constructed.

First, recall that the energy (\ref{j2.1}) is Dirichlet such that, when the range of $\phi$ is the full space $\bfR^3$
rather than $S^2$, a critical point of the energy is harmonic, $\Delta\phi=0$. However, since the range of $\phi$ is confined to
$S^2$, a critical point of (\ref{j2.1}) satisfies, instead, the nonlinear equation
\be\label{j2.59}
\Delta\phi-(\phi\cdot[\Delta\phi])\phi=0,
\ee
whose solutions are much more complicated. (Although the ``true harmonic" equation $\Delta\phi=0$ automatically implies (\ref{j2.59}), the finite-energy condition indicates
that all solutions to $\Delta\phi=0$ are trivial, i.e., $\phi=$ constant.)
 Nevertheless, the work of Belavin and Polyakov \cite{BP} establishes the fact that
solutions of (\ref{j2.59}) are all given through the stereographic projection (\ref{j2.8}) by the solutions of the equation
\be\label{j2.60}
\pa_1 u\pm\ii\pa_2u=0.
\ee
In other words, solutions of (\ref{j2.59}), away from isolated poles, are all represented by holomorphic or anti-holomorphic functions. Thus, the term ``harmonic" is well justified.

Next, in the gauged harmonic map model governed by the energy (\ref{j2.9}) discussed in \S 2.1, the BPS equations consist of the vortex
equation (\ref{jF}) and a ``holomorphic equation" which reads \cite{Sch1,Sch2}
\be\label{j2.61}
D_1 u\pm\ii D_2 u=0,
\ee
thereby replacing the conventional derivatives in  (\ref{j2.60}) by gauge-covariant derivatives.  The equation (\ref{j2.61}) implies that,
away from isolated poles, $u$ is indeed holomorphic or anti-holomorphic up to a smooth multiple \cite{JT}. In other words, we arrive at a similar ``harmonic representation" of the
field configurations through the ``gauged holomorphic equation" (\ref{j2.61}).

Finally, in the context of the dually gauged theory studied in \S 2.2, the single gauged holomorphic equation (\ref{j2.61}) is now replaced by
a pair of gauged holomorphic equations, (\ref{16}) and (\ref{17}), which govern $p,q$, a pair of gauged holomorphic
or anti-holomorphic sections.

Thus we have seen that in gauged models the configuration maps all lie in the category of harmonic maps with well-exhibited
analytic features.

\medskip

In the subsequent sections, we will construct solutions of the gauged harmonic map \cite{Sch1,Sch2,SSY,Y1,Y2} and impurity \cite{TW}, as well as \cite{HY}, inspired
BPS equations (\ref{16})--(\ref{19}) realizing a prescribed distribution of vortices and antivortices, represented as zeros and poles of $q,p$ in (\ref{20}) and (\ref{21}), respectively.

\section{Existence and uniqueness theorems for coexisting vortices and antivortices}

In this section, we state our existence and uniqueness theorems for vortices and antivortices over a compact surface and on $\bfR^2$. We then present
the system of nonlinear elliptic partial differential equations descending from the BPS equations (\ref{16})--(\ref{19}) which govern such vortices and antivortices.

We first consider the equations over a compact Riemann surface  $S$ equipped with a Riemannian metric $g$ and use $\dd\Omega_g$ to denote the
associated canonical surface element.

\begin{theorem}\label{th1}
 Consider the BPS equations \eqref{16}--\eqref{19} of the energy density  \eqref{2}
formulated over a complex Hermitian line bundle $L$ over a compact Riemann surface $S$ with canonical total area
$|S|$ governing two connection 1-forms $\hat{A},\tilde{A}$ and two cross sections $q,p$.  For the prescribed sets of zeros and poles for the fields $q$ and $p$
given respectively in \eqref{20} and \eqref{21}, these coupled equations  admit a solution with
these sets of zeros and poles, if and only if
\ber
&& \max\left\{\big|(c^2+d^2)(N_1-P_1)-(ac+bd)(N_2-P_2)\big|,\right.\nn\\
&&\quad\quad\quad\left.\big|(a^2+b^2)(N_2-P_2)-(ac+bd)(N_1-P_1)\big|\right\}
 <(ad-bc)^2\,\frac{|S|}{\pi}. \label{e.5}
\eer
 Moreover, modulo gauge transformations, such a solution { is unique and} carries  quantized magnetic charges (\ref{2.64}) and (\ref{2.65}), and minimum
energy of the form
\be\label{3.2}
E= \int_S{\cal H}\,{\rm\dd}\Omega_g=4\pi (N_1+N_2+P_1+P_2),
\ee
which is seen to be stratified topologically by the Chern and Thom classes of the line bundle $L$ and its dual respectively.
In particular, in terms of energy, zeros (vortices) and poles (antivortices) of $q$ and $p$ contribute equally.
\end{theorem}

 We next consider the equations \eqref{16}--\eqref{19}  over $\mathbb{R}^2$. We have the following theorem.

\begin{theorem}\label{th2}
 For the BPS equations \eqref{16}--\eqref{19} in $\hat{A},\tilde{A}$ and $q,p$ over $\mathbb{R}^2$ with the
  prescribed sets of zeros and poles for the fields $q$ and $p$
given respectively in \eqref{20} and \eqref{21}, modulo gauge transformations, there is a unique finite-energy solution realizing
these sets of zeros and poles, which represent vortices and antivortices of opposite magnetic charges.  Moreover, such a solution enjoys the sharp exponential decay estimates
\be
1-|q|^2, 1-|p|^2, |Dq|^2,|Dp|^2,|\hF|,|\tF|=\mathrm{O}(\re^{-\sigma(1-\vep)|x|}),\quad \mbox{for $|x|$ large},
\ee
where $\vep>0$ is arbitrarily small and
\be
\sigma^2=8\left(a^2+b^2+c^2+d^2-\sqrt{(a^2+b^2+c^2+d^2)^2-4(ad-bc)^2}\right),
\ee
and carries  quantized magnetic charges (\ref{2.64}) and (\ref{2.65}), now evaluated over $\bfR^2$, and minimum
energy of the form
\be\label{xE}
E=\int_{\bfR^2}{\cal H}\,\dd x= 4\pi (N_1+N_2+P_1+P_2),
\ee
given again in terms of the total numbers of vortices (zeros of $q,p$) and antivortices (poles of $q,p$) only.
\end{theorem}

The above two theorems will be obtained from  Theorems \ref{thm1} and \ref{thm2}, respectively, in the following two sections, except the energy formulas
\eqref{3.2} and \eqref{xE}. In fact, as in \cite{SSY}, we may check that the Thom classes take the values
\be
\frac1{4\pi}\int_S K(q)= P_1,\quad\frac1{4\pi}\int_S K(p)=P_2,
\ee
similarly over $\bfR^2$, which with the help of the quantized integrals in  Theorems \ref{thm1} and \ref{thm2} lead to the derivations of \eqref{3.2} and \eqref{xE},
respectively.

We now proceed to deduce the governing {nonlinear} elliptic partial differential equations of our problem.
To this end, use $\Delta$ to denote the Laplace--Beltrami operator on $S$ induced from the metric $g=(g_{ij})$:
\be
\Delta f=\frac1{\sqrt{|g|}}\pa_i (g^{ij}\sqrt{|g|}\pa_j f),\quad |g|=\det(g_{ij}).
\ee
Then, away from the zeros and poles of $q,p$, we can resolve (\ref{16}), (\ref{17}) to obtain, the relations
\ber
a*\hF+b*\tF&=&-\frac12\Delta \ln|q|^2,\\
c*\hF+d*\tF&=&-\frac12\Delta\ln|p|^2.
\eer
In view of these relations and (\ref{18})--(\ref{21}), we see that $u_1=\ln|q|^2,u_2=\ln|p|^2$ satisfy the equations
\bea
\Delta u_1&=&4(a^2+b^2)\left(\frac{\re^{u_1}-1}{\re^{u_1}+1}\right)+4(ac+bd)\left(\frac{\re^{u_2}-1}{\re^{u_2}+1}\right)\nn\\
&&+4\pi\left(\sum_{s=1}^{N_1}{\delta}_{z'_{1,s}}-\sum_{s=1}^{P_1}\delta_{z''_{1,s}}\right),\quad\quad \label{eq1b}\\
\Delta u_2&=&4(ac+bd)\left(\frac{\re^{u_1}-1}{\re^{u_1}+1}\right)+4(c^2+d^2)\left(\frac{\re^{u_2}-1}{\re^{u_2}+1}\right)\nn\\
&&+4\pi\left(\sum_{s=1}^{N_2}{\delta}_{z'_{2,s}}-\sum_{s=1}^{P_2}\delta_{z''_{2,s}}\right),\quad\quad\,\label{eq1a}
\eea
which govern two species of vortices and antivortices at the prescribed locations at $z'_{i,s'_i}$ and $z''_{i,s''_i}$  with $s'=1,\dots,N_i$ and  $s''_i=
1,\dots,P_i$, $i=1,2$, respectively.

Note again that, in the limit $\re^{u_1},\re^{u_2} \to 1$  in the denominators of the nonlinearities and absence of poles, these equations reduce to those in \cite{HY2} for the generalized Abelian Higgs equations,
inspired by the work of Tong and Wong \cite{TW}. Such a property is elegant.

\section{Vortices and antivortices on a compact surface}

In this section we  study  the equations \eqref{16}--\eqref{19} over a compact Riemann surface $(S, g)$, which have been reduced to   the
elliptic equations \eqref{eq1a}--\eqref{eq1b}. To prove Theorem \ref{th1},  we consider the following more general system of equations:
\bea
\Delta u_1&=&a_{11}\frac{\re^{u_1}-1}{\re^{u_1}+1}+a_{12}\frac{\re^{u_2}-1}{\re^{u_2}+1}+4\pi \sum_{s=1}^{N_1}{\delta}_{z'_{1,s}}-4\pi \sum_{s=1}^{P_1}\delta_{z''_{1,s}},\label{torus1a}\\
 \Delta u_2& =&a_{21}\frac{\re^{u_1}-1}{\re^{u_1}+1}+a_{22}\frac{\re^{u_2}-1}{\re^{u_2}+1}+4\pi \sum_{s=1}^{N_2}{\delta}_{z'_{2,s}}-4\pi \sum_{s=1}^{P_2}\delta_{z''_{2,s}},\label{torus1b}
\eea
 where $a_{ij}$   are constants satisfying $a_{11}>0, a_{12}a_{21}>0$,
and
\be
|a|\equiv\det(a_{ij})=a_{11}a_{22}-a_{12}a_{21}>0.
\ee
Note that, for greater generality of our study, we will not assume symmetry for the coefficient matrix $(a_{ij})$.

Let $u_{0,1},u_{0,2}$ be  solutions of
\bea
 \Delta u_{0,1}&=&4\pi \sum_{s=1}^{N_1}{\delta}_{z'_{1,s}}-4\pi \sum_{s=1}^{P_1}\delta_{z''_{1,s}}+\frac{4\pi}{|S|}(P_1-N_1),\\
 \Delta u_{0,2}&=&4\pi \sum_{s=1}^{N_2}{\delta}_{z'_{2,s}}-4\pi \sum_{s=1}^{P_2}\delta_{z''_{2,s}}+\frac{4\pi}{|S|}(P_2-N_2),
\eea
with $\int_{S}u_{0,1}\ud\Omega_g=\int_{S}u_{0,2}\ud\Omega_g=0$ (cf. \cite{Aubin}). Set $u_i=v_i+u_{0,i}$, $i=1,2$, in (\ref{torus1a})--(\ref{torus1b}). Then $(v_1,v_2)$ solves
\bea
\Delta v_1&=&a_{11}\frac{\re^{v_1+u_{0,1}}-1}{\re^{v_1+u_{0,1}}+1}+a_{12}\frac{\re^{v_2+u_{0,2}}-1}{\re^{v_2+u_{0,2}}+1}+\frac{4\pi}{|S|}(N_1-P_1),\label{torus3a}\\
\Delta v_2 &=&a_{21}\frac{\re^{v_1+u_{0,1}}-1}{\re^{v_1+u_{0,1}}+1}+a_{22}\frac{\re^{v_2+u_{0,2}}-1}{\re^{v_2+u_{0,2}}+1}+\frac{4\pi}{|S|}(N_2-P_2).\label{torus3b}
\eea
Consider the following functional over $H^1(S)\times H^1(S)$:
\bea\label{fun1}
J(v_1,v_2)&=&\frac{a_{22}}2\int_S|\nabla v_1|^2\ud\Omega_g-a_{12}\int_S\nabla v_1\cdot\nabla v_2\ud\Omega_g+\frac{a_{12}a_{11}}{2a_{21}}\int_S |\nabla v_2|^2\ud\Omega_g\nn\\
&&+ |a|\int_S\left(\ln (\re^{v_1+u_{0,1}}+1)+\ln(\re^{-(u_{0,1}+v_1)}+1)\right)\ud\Omega_g\nn\\
&&+\frac{a_{12}|a|}{a_{21}}\int_S\left(\ln (\re^{u_{0,2}+v_2}+1)+\ln(\re^{-(u_{0,2}+v_2)}+1)\right)\ud\Omega_g\nn\\
&&+\frac{4\pi}{|S|}\left(a_{22}(N_1-P_1)-a_{12}(N_2-P_2)\right)\int_S v_1\ud\Omega_g\nn\\
&&+\frac{4\pi}{|S|}\frac{a_{12}}{a_{21}}\left(a_{11}(N_2-P_2)-a_{21}(N_1-P_1)\right)\int_S v_2\ud \Omega_g.
\eea
  In this section we use the following notation
 \ber
 \int_S|\nabla v|^2\ud \Omega_g\equiv  \int_Sg^{jk}\partial_jv\partial_kv\ud \Omega_g, \quad
 \int_S\nabla v_1\cdot\nabla v_2\ud\Omega_g\equiv  \int_Sg^{jk}\partial_jv_1\partial_kv_2\ud \Omega_g.\label{nt1}
 \eer
Then it is clear that \eqref{torus3a} and \eqref{torus3b} are the Euler--Lagrange equations of functional \eqref{fun1}.
\begin{theorem}\label{thm1}
Suppose $a_{11}, a_{12}a_{21},$ and $ |a|>0$. Then the system of equations
\eqref{torus3a}--\eqref{torus3b} admits a solution $(v_1,v_2)\in H^1(S)\times H^1(S)$ if and only if
\ber
\max\Big\{\big|a_{22}(N_1-P_1)-a_{12}(N_2-P_2)\big|,\big|a_{11}(N_2-P_2)-a_{21}(N_1-P_1)\big|\Big\}
<\frac{|a||S|}{4\pi}. \label{e.5}
\eer
Moreover, if the solution exists, it is unique and there hold   the following quantized integrals
\ber
 \int_S\frac{\re^{u_1}-1}{\re^{u_1}+1}\ud\Omega_g=\frac{4\pi}{|a|}(a_{22}[P_1-N_1]-a_{12}[P_2-N_2]), \label{e.5a}\\
 \int_S\frac{\re^{u_2}-1}{\re^{u_2}+1}\ud\Omega_g=\frac{4\pi}{|a|}(a_{11}[P_2-N_2]-a_{21}[P_1-N_1]).\label{e.5b}
\eer
\end{theorem}
\begin{proof}
We first prove that \eqref{e.5} is a necessary condition. If there is a solution $(v_1, v_2)$ of \eqref{torus3a}--\eqref{torus3b},  then after a direct computation we  have
\bea
 \Delta(a_{22}v_1-a_{12}v_2)&=&|a|\,\frac{\re^{u_{0,1}+v_1}-1}{\re^{u_{0,1}+v_1}+1}\nn\\
&&+\frac{4\pi}{|S|}\left(a_{22}[N_1-P_1]-a_{12}[N_2-P_2]\right),\label{torus4a}\\
\Delta (a_{11}v_2-a_{21}v_1)&=&|a|\,\frac{\re^{u_{0,2}+v_2}-1}{\re^{u_{0,2}+v_2}+1}\nn\\
&&+\frac{4\pi}{|S|}\left(a_{11}[N_2-P_2]-a_{21}[N_1-P_1]\right).\label{torus4b}
\eea

Then the quantized integrals \eqref{e.5a}--\eqref{e.5b} follow from  a direct integration of \eqref{torus4a}--\eqref{torus4b} over $S$. Noting the elementary inequality $\left|\frac{e^t-1}{e^t+1}\right|<1$ for any $t\in \mathbb{R}$, we can get the necessity of \eqref{e.5} from  \eqref{e.5a}--\eqref{e.5b}.


We now turn to the proof of sufficiency of  \eqref{e.5}.
Notice the elementary inequality
\begin{equation}\label{t}
\ln(1+\re^t)+\ln(1+\re^{-t})\ge |t|, \quad \forall t\in \mathbb{R},
\end{equation}
which implies
\be
\int_S\left(\ln(1+\re^{u_{0,i}+v_{i}})+\ln(1+\re^{-(u_{0,i}+v_{i})})\right)\ud\Omega_g
\ge\int_S |v_{i}|\ud\Omega_g-\int_S|u_{0,i}|\ud\Omega_g,\, i=1,2. \label{i7}
\ee

Then by \eqref{e.5} and \eqref{i7},  there exists  positive constants $C_0, C_1$ such that
\ber
J(v_1,v_2)&\ge&
 \frac {a_{22}}2\int_S|\nabla v_1|^2\ud\Omega_g-a_{12}\int_S\nabla v_1\cdot\nabla v_2\ud x+\frac{a_{11}a_{12}}{2a_{21}}\int_S |\nabla v_2|^2\ud\Omega_g\nn\\
 &&+|a|\int_S\big[\ln (\re^{v_1+u_{0,1}}+1)+\ln(\re^{-(u_{0,1}+v_1)}+1)\big]\ud \Omega_g\nn\\
 &&-\frac{4\pi}{|S|}\big[a_{22}(N_1-P_1)-a_{12}(N_2-P_2)\big]\int_S|v_1|\ud \Omega_g\nn\\
 &&+\frac{a_{12}|a|}{a_{21}}\int_S\big[\ln (\re^{u_{0,2}+v_2}+1)+\ln(\re^{-(u_{0,2}+v_2)}+1)\big]\ud\Omega_g\nn\\
 &&-\frac{4\pi}{|S|}\big[a_{11}(N_2-P_2)-a_{21}(N_1-P_1)\big]\int_S|v_2|\ud \Omega_g\nn\\
 &\ge&\frac {a_{22}}2\int_S|\nabla v_1|^2\ud\Omega_g-a_{12}\int_S\nabla v_1\cdot\nabla v_2\ud x+\frac{a_{11}a_{12}}{2a_{21}}\int_S |\nabla v_2|^2\ud\Omega_g\nn\\
 &&+\left(|a|-\frac{4\pi}{|S|}\big[a_{22}(N_1-P_1)-a_{12}(N_2-P_2)\big]\right)\int_S|v_1|\ud \Omega_g\nn\\
 &&+\frac{a_{12}}{a_{21}}\left(|a|-\frac{4\pi}{|S|}\big[a_{11}(N_2-P_2)-a_{21}(N_1-P_1)\big]\right)\int_S|v_2|\ud\Omega_g\nn\\
&&-|a|\int_S|u_{0,1}|\ud\Omega_g -\frac{a_{12}|a|}{a_{21}}\int_S|u_{0,2}|\ud\Omega_g\nn\\
&\ge& C_0 \int_S\left\{ |\nabla v_1|^2+|\nabla v_2|^2+|v_1|+|v_2|\right\}\ud\Omega_g-C_1,\label{4.17}
\eer
which says that  the functional $J(v_1, v_2)$ is bounded from below. Using  \eqref{4.17} and Poincar\'{e} inequality, we observe also that the functional is coercive in a well understood sense.
Then the minimization problem
\begin{equation}
\eta_0\equiv\inf\{J(v_1,v_2)\,|\,(v_1, v_2)\in H^1(S)\times H^1(S)\}
\end{equation}
is well posed.

Consider a minimizing sequence $\{(v_{1n},v_{2n})\}\subset H^1(S)\times H^1(S)$. Then we can easily see that, for some positive constant $C>0$,
\begin{equation}
\int_S\big(|\nabla v_{1n}|^2+|\nabla v_{2n}|^2\big)\ud\Omega_g\le C.
\end{equation}
Set $v_{in}=w_{in}+c_{in}$ with $\int_S w_{in}=0, c_{in}\in \mathbb{R}, i=1,2$. Then by the Moser--Trudinger
inequality \cite{Aubin,F}
\be
\int_S \re^f\,\ud\Omega_g\leq C\exp\left(\frac1{16\pi}\int_S |\nabla f|^2\,\ud\Omega_g\right),\quad f\in H^1(S),\quad\int_S f\,\ud\Omega_g=0,
\ee
where $C>0$ is a constant, and the Poincar\'{e} inequality, we have
\begin{equation}
\int_S \left(w_{1n}^2+w_{2n}^2\right)\ud\Omega_g\le C.
\end{equation}

From \eqref{4.17}  we can see that  $|c_{in}|\le C<\infty, i=1, 2$.  This implies $(v_{1n},v_{2n})$ is weakly compact in $H^1(S)\times H^1(S)$.
Then, up to a subsequence, there exists $(v_1^*, v_2^*)\in H^1(S)\times H^1(S)$ such that $(v_{1n},v_{2n})\to (v_1^*, v_2^*)$  weakly as $n\to \infty$.
Consequently, $(v_1^*, v_2^*)$ is a minimizer of the functional  $J$, which gives a weak  solution of the  system \eqref{torus3a}--\eqref{torus3b}. A standard bootstrap argument then shows that it is also a smooth solution.

For the uniqueness part, we just need to show the minimizer is unique. It is a  consequence of the fact that the functional $J(v_1, v_2)$  is strictly convex, which can be checked straightforwardly, thereby completing the proof.
\end{proof}

\section{Planar solution: proof of existence by minimization}

To prove Theorem \ref{th2}, we consider (\ref{torus1a})--(\ref{torus1b})  
over $ \mathbb R^2$
subject to the boundary condition
\be
 u_i\to 0\quad \text{as}\quad |x|\to \infty, \quad i=1, 2. \label{se2}
\ee

We have the following main existence and uniqueness theorem.

\begin{theorem}\label{thm2}
Suppose $a_{11}, a_{12}a_{21},$ and $|a|>0$. Then the system of equations \eqref{torus1a}--\eqref{torus1b}
over $\bfR^2$ admits a {unique} solution $(u_1,u_2)$ satisfying \eqref{se2}.
Moreover, the solution enjoys the following sharp decay estimates
 \ber
  u_1^2+u_2^2\le C(\vep)\re^{-\sqrt{\lambda_0}(1-\vep)|x|}, \label{se3}\\
   |\nabla u_1|^2+|\nabla u_2|^2\le C(\vep)\re^{-\sqrt{\lambda_0}(1-\vep)|x|},\label{se4}
 \eer
for $|x|$ being sufficiently large,   where $\vep\in(0, 1)$ is arbitrarily small,  $C(\vep)>0$ a corresponding constant, and $\lm_0>0$  given by the formula
\be
\lm_0\equiv\frac12\min\left\{1,\frac{a_{21}}{a_{12}}\right\}\left(a_{11}+\frac{a_{12}}{a_{21}}a_{22}-\sqrt{\Big(a_{11}+\frac{a_{12}}{a_{21}}a_{22}\Big)^2-4\frac{a_{12}}{a_{21}}|a|}\right).
\ee
Furthermore, there hold the  quantized integrals
\ber
 \itr\frac{\re^{u_1}-1}{\re^{u_1}+1}\ud x=\frac{4\pi}{|a|}(a_{22}[P_1-N_1]-a_{12}[P_2-N_2]),\label{se5} \\
 \itr\frac{\re^{u_2}-1}{\re^{u_2}+1}\ud x=\frac{4\pi}{|a|}(a_{11}[P_2-N_2]-a_{21}[P_1-N_1]).\label{se6}
\eer
\end{theorem}


We now proceed to prove the existence and uniqueness of a solution by the method of calculus of variations. The stated decay estimates and
quantized integrals of a solution will be established in the next section.

Let $u_{0,i}$ be given by
\begin{equation}
u_{0,i}=-\frac 12\sum_{j=1}^{N_i}\ln\big(1+\lambda|x-z'_{i,j}|^{-4}\big)+\frac 12\sum_{j=1}^{P_i}\ln \big(1+\lambda|x-z''_{i,j}|^{-4}\big),\quad i=1,2.\label{se7}
\end{equation}
By a direct computation, we have
\be\label{back1}
\Delta u_{0,i}=4\pi\sum_{j=1}^{N_i}\delta_{z'_{i,j}}-4\pi\sum_{j=1}^{P_i}\delta_{z''_{i,j}}-f_i,
\ee
where
\be
f_i\equiv8\sum_{j=1}^{N_i}\frac{\lambda|x-z'_{i,j}|^2}{(\lambda+|x-z'_{i,j}|^4)^2}-8\sum_{j=1}^{P_i}\frac{\lambda |x-z''_{i,j}|^2}{(\lambda+|x-z''_{i,j}|^4)^2},\quad i=1,2.
\ee
Note that
\begin{equation}
\frac{\lambda|x|^2}{(\lambda+|x|^4)^2}\le \begin{cases}
\displaystyle\frac{\lambda^{\frac 53}}{\lambda^2}\le \lambda^{-\frac 13},\quad \text{for}\quad |x|^3\le \lambda,\\
\displaystyle\frac{\lambda}{|x|^6}\le \lambda^{-1},\quad \text{for}\quad |x|^3\ge \lambda,
\end{cases}
\end{equation}
which  implies $f_i\rightarrow 0$ as $\lambda\rightarrow \infty$ uniformly for $i=1,2$. Fix $\lambda$ large to be determined later. Also from the above construction, we know $u_{0,i}\in L^1(\mathbb R^2)\cap L^2(\bfR^2)$, $i=1,2$.

Set $v_i=u_i-u_{0,i}$, $i=1,2$.  Then we reduce \eqref{torus1a}--\eqref{torus1b} and \eqref{se2} into
\bea
 \Delta v_1&=&a_{11}\frac{\re^{u_{0,1}+v_1}-1}{\re^{u_{0,1}+v_1}+1}+a_{12}\frac{\re^{u_{0,2}+v_2}-1}{\re^{u_{0,2}+v_2}+1}+f_1,\label{space2a}\\
 \Delta v_2&=&a_{21}\frac{\re^{u_{0,1}+v_1}-1}{\re^{u_{0,1}+v_1}+1}+a_{22}\frac{\re^{u_{0,2}+v_2}-1}{\re^{u_{0,2}+v_2}+1}+f_2,\label{space2b}
\eea
over $\mathbb R^2$
and
\be
v_i\to 0\quad    \text{as}\quad |x|\to \infty, \quad i=1, 2,
\ee
respectively.


We now consider the  functional $J(v_1,v_2)$ for $(v_1,v_2)\in H^1(\mathbb R^2)\times H^1(\mathbb R^2)$ as follows:
\ber\label{fun2}
J(v_1,v_2)&\equiv&\frac{a_{22}}2\int_{\mathbb R^2}|\nabla v_1|^2\ud x-a_{12}\int_{\mathbb R^2}\nabla v_1\cdot\nabla v_2\ud x
+\frac{a_{11}a_{12}}{2a_{21}}\int_{\mathbb R^2}|\nabla v_2|^2\ud x
\nn\\
&&
+|a|\int_{\mathbb R^2}\ln\left(\frac{\re^{u_{0,1}+v_1}+\re^{-(u_{0,1}+v_1)}+2}{4}\right)\ud x
\nn\\
&&
+\frac{a_{12}}{a_{21}}|a|\int_{\mathbb R^2}\ln\left(\frac{\re^{u_{0,2}+v_2}+\re^{-(u_{0,2}+v_2)}+2}{4}\right)\ud x\nn\\
&&+\int_{\mathbb R^2}\tilde{f}_1v_1\ud x+\frac{a_{12}}{a_{21}}\int_{\mathbb R^2}\tilde{f}_2v_2\ud x,
\eer
where and in the sequel we use the notation
\ber
 \tilde{f}_1\equiv a_{22}f_1-a_{12}f_2, \quad  \tilde{f}_2\equiv a_{11}f_2-a_{21}f_1.
\eer
Now we show that the functional  $J(v_1,v_2)$ is well defined.  We split  $\mathbb R^2$ into two parts:
\bea
\mathbb R^2&=&\left\{\re^{u_{0,i}+v_i}+\re^{-(u_{0,i}+v_i)}\le 4\right\}\cup\left\{\re^{u_{0,i}+v_i}+\re^{-(u_{0,i}+v_i)}> 4\right\}\nn\\
&\equiv& A_i\cup \mathbb (\bfR^2\backslash A_i),\quad  i=1, 2.\label{517}
\eea
Then we have
\ber
0&\le &\int_{A_i}\ln\left(\frac{\re^{u_{0,i}+v_i}+\re^{-(u_{0,i}+v_i)}+2}{4}\right)\ud x\nn\\
&=&\int_{A_i}\ln\left(1+\frac12\left[\frac{\re^{u_{0,i}+v_i}+\re^{-(u_{0,i}+v_i)}}{2}-1\right]\right)\ud x\nn\\
&\le & \int_{A_i} \frac12\left(\frac{\re^{u_{0,i}+v_i}+\re^{-(u_{0,i}+v_i)}}{2}-1\right) \ud x\nn\\
&\le & \frac12\int_{A_i} (u_{0,i}+v_i)^2\ud x\nn\\
&\le&\int_{\bfR^2} u_{0,i}^2\ud x+\int_{\mathbb R^2}v_i^2\ud x.\label{518}
\eer
On the other hand, from $2\re^{|t|}\geq \re^t+\re^{-t}$, we obtain $|u_{0,i}+v_i|> \ln2$ in $\mathbb R^2\backslash A_i$.  Thus,
 we have
\ber
0&\le &\int_{\mathbb R^2\backslash A_i}\ln\left(\frac{\re^{u_{0,i}+v_i}+\re^{-(u_{0,i}+v_i)}+2}{4}\right)\ud x\nn\\
 &\le&\int_{\mathbb R^2\backslash A_i}\ln\left(\re^{u_{0,i}+v_i}+\re^{-(u_{0,i}+v_i)}\right)\ud x\nn\\
 &\le&\int_{\bfR^2\backslash A_i}\ln \left(2\re^{|u_{0,i}+v_i|}\right)\ud x\nn\\
&\le&\int_{\bfR^2\backslash A_i}\ln\left(\re^{2\ln|u_{0,i}+v_i|}\right)\,\ud x\nn\\
 &=&\int_{\bfR^2\backslash A_i} (u_{0,i}+v_i)^2\ud x\nn\\
&\le& 2\int_{\bfR^2} u_{0,i}^2\ud x+2\int_{\mathbb R^2}v_i^2\ud x.\label{519}
 \eer
The bounds (\ref{518})--(\ref{519}) establish that
 the functional  $J$ is well defined. In fact, in line of these estimates, it becomes clear to show that $J$ is $C^1$
over $H^1(\bfR^2)\times H^1(\bfR^2)$. Thus, as in the compact case,  it is seen that the equations
\eqref{space2a}--\eqref{space2b} are the Euler--Lagrange equations of the functional $J$ given in (\ref{fun2}).
Note also that the functional $J$ is strictly convex. So $J$ has at most one critical point in $H^1(\bfR^2)\times H^1(\bfR^2)$.
In particular the uniqueness of a solution of \eqref{space2a}--\eqref{space2b} in $H^1(\bfR^2)\times H^1(\bfR^2)$
follows.

Next, we aim to get some lower control of the functional.   For convenience, we modify the splitting (\ref{517}) slightly:
\bea
\mathbb R^2&=&\left\{\re^{u_{0,i}+v_i}+\re^{-(u_{0,i}+v_i)}\le 5\right\}\cup\left\{\re^{u_{0,i}+v_i}+\re^{-(u_{0,i}+v_i)}> 5\right\}\nn\\
&\equiv& A_i\cup \mathbb (\bfR^2\backslash A_i),\quad  i=1, 2.
\eea
On $A_i$, using Taylor's truncation,  we  have
\ber
&&\int_{A_i}\ln\left(\frac{\re^{u_{0,i}+v_i}+\re^{-(u_{0,i}+v_i)}+2}{4}\right)\ud x\nn\\
&&=\int_{A_i}\ln\left(1+\frac12\left[\frac{\re^{u_{0,i}+v_i}+\re^{-(u_{0,i}+v_i)}}{2}-1\right]\right)\ud x\nn\\
&&\ge\frac12\int_{A_i}\left(\left(\cosh(u_{0,i}+v_i)-1\right)-\frac14\left(\cosh(u_{0,i}+v_i)-1\right)^2\right) \ud x\nn\\
&&\ge\frac5{16}\int_{A_i} \left(\cosh(u_{0,i}+v_i)-1\right)\,\ud x\nn\\
&&\ge \frac5{32}\int_{A_i}(u_{0,i}+v_i)^2\ud x\nn\\
&&\ge \frac5{64}\int_{A_i}v_i^2\ud x-\frac5{32}\int_{\bfR^2} u^2_{0,i}\,\ud x\, \quad i=1, 2.\label{522}
\eer
On the other hand, on $\mathbb R^2\backslash A_i$, we have
\begin{equation}
(\re^{u_{0,i}+v_i}+\re^{-(u_{0,i}+v_i)})^{\theta}\ge 4,\quad \text{for some }\theta\in(0,1),\quad  i=1,2.
\end{equation}
Then we  see  that
\ber
&&\int_{\mathbb R^2\backslash A_i}\ln\left(\frac{\re^{u_{0,i}+v_i}+\re^{-(u_{0,i}+v_i)}+2}{4}\right)\ud x\nn\\
&&\ge(1-\theta)\int_{\mathbb R^2\backslash A_i}\ln\left(\re^{u_{0,i}+v_i}+\re^{-(u_{0,i}+v_i)}\right)\ud x\nn\\
&&\ge (1-\theta)\int_{\mathbb R^2\backslash A_i}\ln \re^{|u_{0,i}+v_i|}\ud x\nn\\
&&\ge C_0\int_{\mathbb R^2\backslash A_i}|v_i|\ud x-{C_1(\lambda)},\label{524}
\eer
where $C_0,C_1(\lambda)$ are some positive constants.

Furthermore, for the last two terms in \eqref{fun2}, we  have
\be\label{525}
 \int_{\mathbb R^2}\tilde{f}_iv_i\ud x
 \ge-\vep\int_{A_i}|v_i|^2\ud x-\frac{1}{4\vep}\int_{A_i}|\tilde{f}_i|^2\ud x-\int_{\mathbb R^2\backslash A_i}|\tilde{f}_iv_i|\ud x,
\ee
where $\vep>0$ is a constant to be determined later.

Combining (\ref{522}), (\ref{524}), and (\ref{525}), and
noting the fact that $f_i\in L^p(\mathbb{R}^2)$ for any $p\ge1$ and  $f_i\rightarrow 0$ as $\lambda\rightarrow \infty$ uniformly for $i=1,2$,
and that  $\vep$ in (\ref{525}) may be chosen to be small, we arrive at
\begin{equation}\label{coe1}
J(v_1,v_2)\ge C_1\sum_{i=1}^2\left(\int_{\mathbb R^2}|\nabla v_i|^2 \ud x+\int_{A_i}v_i^2\ud x+\int_{\mathbb R^2\backslash A_i}|v_i|\ud x\right)-{C_2(\lambda)},
\end{equation}
where $C_1,C_2(\lambda)$ are positive constants independent of $v_1,v_2$. In particular, we see that
the functional $J$ is bounded from below in the space $H^1(\mathbb R^2)\times H^1(\mathbb R^2)$
such that the minimization problem
 \be
\displaystyle \eta_0\equiv\inf\{J(v_1,v_2)\,|\, (v_1,v_2)\in H^1(\mathbb R^2)\times H^1(\mathbb R^2)\}
\ee
 is well-defined, as in the compact surface situation.

   Consider a minimizing sequence $\{(v_{1n},v_{2n})\}\subset C_c^\infty(\mathbb R^2)$. We need to show  $\{(v_{1n},v_{2n})\}$ is a bounded sequence in $H^1(\mathbb R^2)\times H^1(\mathbb R^2)$.
  In  what follows we use $v_{n}$ to denote $v_{in}$ and $A_n$ to denote the corresponding $A_{in}$ ($i=1,2$) defined in
(\ref{coe1}) associated with the pair $(v_{1n},v_{2n})$. From \eqref{coe1}, we have the bound
   \begin{equation}
   \|\nabla v_{n}\|_{L^2(\mathbb R^2)}+\int_{A_{n}}v_{n}^2\ud x+\int_{\mathbb R^2\backslash A_{n}}|v_{n}|\ud x\le C,
   \end{equation}
for some constant $C>0$.
Now, for any fixed $\delta>0$, we have
 \ber
\int_{\{|v_{n}|\ge\delta\}} |v_{n}| \ud x&=&\int_{A_{n}\cap \{|v_{n}|\ge \delta\}}|v_{n}|\ud x+\int_{(\mathbb R^2\backslash A_{n})\cap\{|v_{n}|\ge \delta\}}|v_{n}|\ud x\nn\\
&\le& \frac 1{\delta}\int_{A_{n}\cap \{|v_{n}|\ge \delta\}}v_{n}^2\ud x+\int_{(\mathbb R^2\backslash A_{n})\cap\{|v_{n}|\ge \delta\}}|v_{n}|\ud x\nn\\
&\le&{C}\max\left\{\frac1\delta,1\right\},\label{sp1}
\eer
and
\ber
\int_{\{|v_{n}|\le \delta\}} v_{n}^2 \ud x&=&\int_{A_{n}\cap \{|v_{n}|\le \delta\}}v_{n}^2\ud x+\int_{(\mathbb R^2\backslash A_{n})\cap\{|v_{n}|\le \delta\}}v_{n}^2\ud x\nn\\
&\le& \int_{A_{n}\cap \{|v_{n}|\le \delta\}}v_{n}^2\ud x+\delta \int_{(\mathbb R^2\backslash A_{n})\cap\{|v_{n}|\le \delta\}}|v_{n}|\ud x\nn\\
&\le&C\max\{\delta,1\}.\label{sp2}
\eer

 Now consider the open set $\{v_{n}>\delta\}$. Without loss of generality, we assume $\delta=1$. 
Let
\ber
B_{n}\equiv\{v_{n}>1\}.
\eer
 By \eqref{sp1}, we know $|B_{n}|\le C$. Also since $v_{n}$ is compactly supported,  $B_{n}$ is a bounded open set.
Define
\begin{equation}
\tilde v_{n}\equiv \max\{v_{n}-1,0\}.
\end{equation}
Then, from the above analysis, we have
\be\label{533}
\int_{\mathbb R^2} |\nabla \tilde v_{n}|^2\ud x+\int_{\mathbb R^2} |\tilde v_{n}|\ud x=\int_{B_{n}} |\nabla v_{n}|^2\ud x+\int_{B_{n}} |v_{n}-1|\ud x\le 2C.
\ee

Recall a special case of the Gagliardo--Nirenberg interpolation inequality (cf. Theorem 12.83 in \cite{Leoni}):
\be
\|f\|_{L^r(\mathbb R^2)}\le C(\alpha)\|\nabla f\|_{L^2(\mathbb{R}^2)}^{\alpha}\|f\|_{L^1(\mathbb{R}^2)}^{1-\alpha}, \quad r=\frac{1}{1-\alpha},\quad \forall \alpha\in(0, 1). \label{gni}
\ee
Taking $r=2$ in \eqref{gni} and applying it to $\tilde{v}_n$,  we see that, in view of (\ref{533}), there holds the bound
\begin{equation}
\int_{\mathbb R^2}\tilde v_{n}^2\ud x\le C^2\left(\frac12\right)\|\nabla \tilde v_{n}\|_{L^2(\mathbb{R}^2)}\|\tilde v_{n}\|_{L^1(\mathbb{R}^2)}\leq C_0,
\end{equation}
where $C_0>0$ is a constant independent of $n$. This in turn implies
\begin{equation}\label{sp3}
\int_{B_{n}} v_{n}^2\ud x=\int_{B_{n}}(\tilde v_{n}+1)^2\ud x\le 2\int_{B_{n}}\tilde v_{n}^2\ud x+2|B_{n}|\le C_1,
\end{equation}
where $C_1>0$ is a constant independent of $n$.

Likewise, for  $\tilde B_{n}\equiv\{v_{n}<-1\}$ (say),  there also holds
\begin{equation}\label{sp4}
\int_{\tilde B_{n}} v_{n}^2\ud x\le C_2,
\end{equation}
with an $n$-independent positive constant $C_2$.

From the estimates \eqref{sp2}, \eqref{sp3}, and \eqref{sp4},  we see that  the sequence $\{v_{n}\}$ is bounded in $H^1(\mathbb R^n)$. That is, the minimizing sequence $\{(v_{1n},v_{2n})\}$ is bounded in $H^1(\bfR^2)\times H^1(\bfR^2)$.

Then, taking a subsequence if necessary,  we may assume that  $ v_{in}\rightarrow v_i$ weakly in { $H^1(\mathbb{R}^2)$}, and $v_{in}\rightarrow  v_i$ a.e. in $\mathbb R^2$ for some $ v_i\in H^1(\mathbb{R}^2)$, $i=1,2$.
 Noting that the functional $J$ is $C^1$ in $H^1(\bfR^2)\times H^1(\bfR^2)$  and weakly lower semicontinuous, which is ensured by its 
 convexity, we see that  $J(v_1,v_2)= \eta_0$.
 As a (unique) critical point of $J$ in $H^1(\bfR^2)\times H^1(\bfR^2)$, $(v_1,v_2)$ is a weak solution
 of \eqref{space2a}--\eqref{space2b}.  

 We may check that the right-hand sides of the system of equations \eqref{space2a}--\eqref{space2b}
belong to $L^2(\mathbb{R}^2)$. Then by the standard elliptic
$L^2$-estimates we have $v_i\in H^2(\mathbb{R}^2)$, which implies $v_i\to 0$ as $|x|\to \infty$, $i=1,2$.
Furthermore, a bootstrap argument shows that $(v_1,v_2)$ is a smooth solution of \eqref{space2a}--\eqref{space2b}.

\section{Planar solution:  exponential decay properties and quantized integrals}

Now we first establish the decay estimates of the planar solution at infinity.

Let $R_0$ satisfy
\be
R_0>\max\left\{\max\limits_{i=1,2;1\le j\le N_i}|z'_{i,j}|, \max\limits_{i=1,2;1\le j \le P_i}|z''_{i,j}|\right\}
\ee
and $D_R$ denote a disk in $\bfR^2$ centered at the origin with radius $R>0$.
Then, outside $D_{R_0}$, we may conveniently rewrite \eqref{torus1a}--\eqref{torus1b} as
\be
 \Delta \begin{pmatrix}u_1\\\frac{a_{12}}{a_{21}}u_2\end{pmatrix}=\tilde{A}\begin{pmatrix}\frac{\re^{u_1}-1}{\re^{u_1}+1}\\\frac{\re^{u_2}-1}{\re^{u_2}+1}\end{pmatrix},
\ee
where
\ber
 \tilde{A}\equiv \begin{pmatrix} a_{11}&a_{12}\\a_{12}&\frac{a_{12}}{a_{21}}a_{22}\end{pmatrix} \label{e49}
\eer
is a positive definite matrix whose smaller eigenvalue is
\be
\lambda_1\equiv\frac12\left\{a_{11}+\frac{a_{12}}{a_{21}}a_{22}-\sqrt{\Big(a_{11}+\frac{a_{12}}{a_{21}}a_{22}\Big)^2-4\frac{a_{12}}{a_{21}}|a|}\right\}>0.
\ee
Let $w\equiv u_1^2+\frac{a_{12}}{a_{21}}u_2^2$.  Noting that $u_1, u_2$ vanish at infinity,  after a direct computation, we have
\ber
 \Delta w&\ge& 2\left(u_1\Delta u_1+\frac{a_{12}}{a_{21}}u_2\Delta u_2\right)\nn\\
 &=&2(u_1, u_2)\tilde{A}\begin{pmatrix}\frac{\re^{u_1}-1}{\re^{u_1}+1}\\\frac{\re^{u_2}-1}{\re^{u_2}+1}\end{pmatrix}\nn\\
 &=&(u_1, u_2)\tilde{A}\begin{pmatrix}u_1\\u_2\end{pmatrix}-2(u_1, u_2)\tilde{A}\begin{pmatrix}\frac12u_1-\frac{\re^{u_1}-1}{\re^{u_1}+1}\\\frac12{u_2}-\frac{\re^{u_2}-1}{\re^{u_2}+1}\end{pmatrix}\nn\\\nn\\
 &\ge&\lambda_1(u_1^2+u_2^2)-g(x)(u_1^2+u_2^2)
\eer
where  $g(x)$ is a function vanishing at infinity.
Then, for any $\vep\in(0,1)$, there exists an $R_\vep>0$ such that
 \ber
 \Delta w\ge\lambda_1\left(1-\frac\vep2\right)(u_1^2+u_2^2)\ge \lambda_0\left(1-\frac\vep2\right)w \quad \text{as} \quad |x|>R_\vep,\label{e50}
 \eer
 where
 \be
 \lambda_0\equiv\lambda_1\min\Big\{1, \frac{a_{21}}{a_{12}}\Big\}\label{lm0}
 \ee
Hence by  \eqref{e50} and a comparison function argument we infer that,  for any $\vep\in(0, 1)$, there exists a constant $C(\vep)>0$
such that
\ber
 w&\le&C(\vep)\re^{-\sqrt{\lambda_0}(1-\vep)|x|}\quad  \text{as} \quad x>R_\vep.\label{e50a}
\eer

Next, let $\partial$ denote any of the two derivatives, $\partial_1$ and $\partial_2$.
Thus, when  $|x|>R_0$, we have
\ber
 \Delta \begin{pmatrix}\partial u_1\\\partial u_2\end{pmatrix}=2A\begin{pmatrix}\frac{\re^{u_1}\partial {u_1}}{(\re^{u_1}+1)^2}\\\frac{\re^{u_2}\partial {u_2}}{(\re^{u_2}+1)^2}\end{pmatrix}.\label{e51}
\eer
 Noting  that the right-hand side of \eqref{e51} belongs to $L^2(\mathbb{R}^2\setminus D_{R_0})$ and using the elliptic  $L^2$-estimate there,  we obtain
 $\partial u_1, \partial u_2\in W^{2,2}(\mathbb{R}^2\setminus D_{R_0})$, which implies in particular that  $\partial u_1$ and $\partial u_2$ vanish at infinity.

Rewrite \eqref{e51} again as before in the familiar form:
 \ber
 \Delta \begin{pmatrix}\partial u_1\\ \frac{a_{12}}{a_{21}}\partial u_2\end{pmatrix}=2\tilde{A}\begin{pmatrix}\frac{\re^{u_1}\partial u_1}{(\re^{u_1}+1)^2}\\\frac{\re^{u_2}\partial u_2}{(\re^{u_2}+1)^2}\end{pmatrix}, \quad x>R_0,\label{e52}
\eer
where $\tilde{A}$ is given in \eqref{e49}.
Set $W\equiv(\partial u_1)^2+\frac{a_{12}}{a_{21}}(\partial u_2)^2$. Then   we have
 \ber
 \Delta W &\ge& 2\left(\partial u_1\Delta\partial u_1+\frac{a_{12}}{a_{21}}\partial u_2\Delta\partial u_2\right)\nn\\
  &=&4(\partial u_1, \partial u_2)\tilde{A}\begin{pmatrix}\frac{\re^u\partial u_1}{(\re^{u_1}+1)^2}\\\frac{\re^{u_2}\partial {u_2}}{(\re^{u_2}+1)^2}\end{pmatrix}\nn\\
  &=& (\partial u_1, \partial u_2)\tilde{A}\begin{pmatrix}\partial u_1\\\partial u_2\end{pmatrix}  -4(\partial u_1, \partial u_2)A\begin{pmatrix}\left(\frac14-\frac{\re^{u_2}}{(\re^{u_1}+1)^2}\right)\partial u_1\\\left(\frac14-\frac{\re^{u_2}}{(\re^{u_2}+1)^2}\right)\partial u_2\end{pmatrix}\nn\\
  &\ge&\lambda_1((\partial u_1)^2+(\partial u_2)^2)-h(x)((\partial u_1)^2+(\partial u_2)^2), \quad x>R_0,
 \eer
where $h(x)$ is a function vanishing at infinity.
Thus, similar as in getting \eqref{e50a}, we  deduce that, for any $\vep\in(0, 1)$, there exists a constant $C(\vep)>0$
such that
\ber
 W&\le&C(\vep)\re^{-\sqrt{\lambda_0}(1-\vep)|x|}\quad  \text{as} \quad x>R_\vep, \label{e53}
\eer
where $\lambda_0$ is defined as in \eqref{lm0}. Therefore, the desired estimates \eqref{se3}--\eqref{se4} follows from \eqref{e50a} and \eqref{e53}.

Finally we derive the quantized integrals. To proceed, we note that, in view of the properties of $\nabla u_i$ stated in \eqref{se4} and of $u_{0,i}$
given in \eqref{se7} ($i=1,2$),  we have
${|\nabla v_i|=O(|x|^{-5})}$ (say) as $|x|\to \infty$, $i=1,2$.  As a consequence, the divergence theorem then leads to
\be\label{613}
\int_{\bfR^2}\Delta v_i\,\ud x=\lim_{R\to\infty}\int_{D_R}\Delta v_i\,\ud x=0,\quad i=1,2.
\ee
Besides, a direct integration gives us
\ber\label{614}
 \itr f_i\,\ud x=4\pi(P_i-N_{i}), \quad i=1,2.
\eer
 Combining (\ref{613}) and (\ref{614}), we obtain
 \ber
  a_{11}\itr\frac{\re^{u_1}-1}{\re^{u_1}+1}\ud x+a_{12}\itr\frac{\re^{u_2}-1}{\re^{u_2}+1}\ud x=4\pi(P_1-N_1),\\
 a_{21}\itr\frac{\re^{u_1}-1}{\re^{u_1}+1}\ud x+a_{22}\itr\frac{\re^{u_2}-1}{\re^{u_2}+1}\ud x=4\pi(P_2-N_2),
 \eer
from which the anticipated quantized integrals \eqref{se5}--\eqref{se6} then follow.

\section{Summary and comments}

Extended quantum field theory models hosting multiple sectors of the Higgs fields are of wide range of applications including
superconductivity, elementary particles, condensed-matter physics, and cosmology. In these applications, vortices, or vortexlines, often provide useful conceptual constructs and mechanisms for interactions at fundamental levels. Thus, realization
and uncovery of vortices of novel features are of value. In this work, we developed a gauge field theory allowing the
coexistence of vortices and antivortices and established a series of sharp existence and uniqueness theorems for the solutions of the governing equations.

\begin{enumerate}
\item[(i)] Based on the gauged harmonic map model and the product Abelian Higgs theory hosting impurities, an
extended and dually coupled
gauged harmonic map model is presented in which two species of vortices and antivortices coexist and are governed by
 vortex equations of a BPS type. Topologically, the solutions are characterized by two classification classes, namely, the first
Chern class of the defining line bundle and the Thom class of the associated dual bundle. Mathematically, the vortices and
antivortices of solutions are given by the zeros and poles of the cross-sections where two induced magnetic fields represented
by bundle curvatures attain their peaks and valleys.

\item[(ii)] For the vortex equations over a compact surface modeling a doubly periodic lattice structure, an existence and
uniqueness theorem for a solution realizing an {arbitrarily} given prescribed distribution of vortices and antivortices is proved
under a necessary and sufficient condition relating the total numbers of vortices and antivortices and the coupling parameters
involved. This condition is independent of the locations of the vortices but gives an explicit upper bound for the differences
of the numbers of vortices and antivortices in terms of the total area of the hosting surface.

\item[(iii)] For the vortex equations over the full plane, an existence and uniqueness theorem for coexisting vortex and antivortex
solutions is also proved for arbitrary coupling parameters and vortex numbers. The solutions describe spontaneously broken
vacuum symmetry at spatial infinity and are energetically localized configurations. Sharp exponential decay estimates of the solutions are obtained as well.

\item[(iv)] The magnetic fluxes and energies of the solutions over a compact surface and on the full plane are all quantized and
expressed in the terms of total numbers of vortices and antivortices. Specifically, the magnetic fluxes
are determined by the differences of numbers of vortices and antivortices, suggesting the fact that magnetically these vortices
annihilate each other, and the energy on the other hand is given in terms of the sum of the total numbers of all vortices, indicating the fact that
energetically these vortices make equal or indistinguishable contributions as field solitons.
\end{enumerate}

This work opens some future directions to be explored further. For example, it will be interesting to investigate the solutions
realizing an unbroken vacuum symmetry at infinity characterized by the boundary condition $\phi=\psi={\bf n}$ in (\ref{jH0}) or
$q=p=0$ in a slightly modified version of the system of equations \eqref{16}--\eqref{19} at infinity. It will also be interesting to
study the problem of coexisting cosmic strings and antistrings when the model is coupled with the Einstein gravity, especially
the issue how these vortices give rise to localized curvature distribution and how they determine the deficit angle
and geodesic completeness of the induced
gravitational metric at infinity.

In a broader context, this work belongs to the study of coexisting field-theoretical solitons carrying opposite soliton charges, among which one of the most interesting
applications is to use a monopole and antimonopole pair to model a quark and antiquark pair in interaction,  to probe the linear confinement mechanism of quarks, as briefly reviewed in Introduction. However, at the
governing equation levels, there has been no successful construction of solutions realizing a monopole and antimonopole pair, in three-spatial dimensional settings. 
Our work here, on the other hand, is a construction of vortices and antivortices, either paired or unpaired, realizing opposite magnetic charges, in both compact and
noncompact situations, in two-spatial dimensional settings. Hopefully, this lower-dimensional construction will offer useful insight for the investigation in higher-dimensional settings.

\medskip

\medskip

{\bf Acknowledgments.}
 Han was supported by  National Natural Science Foundation of China under Grant 11671120 and  HASTIT (18HASTIT028). Huang was supported by  National Natural Science Foundation of China under Grant 11871160. Yang was partially supported by  National Natural Science Foundation of China under Grant 11471100.

\end{document}